\documentclass[10pt,twoside,leqno]{article}
\usepackage[centertags]{amsmath}
\usepackage{amsthm}
\usepackage{amssymb}
\usepackage{url}
\usepackage{color}
\usepackage[
pdftitle={Lie remarkable ODEs},%
pdfauthor={G. Manno, F. Oliveri, G. Saccomandi, R. Vitolo},%
pdfkeywords={differential equations, point symmetries}]{hyperref}

\newtheorem{theorem}{Theorem}
\newtheorem{corollary}[theorem]{Corollary}

\newtheorem{proposition}[theorem]{Proposition}
\theoremstyle{definition}
\newtheorem{remark}[theorem]{Remark}
\newtheorem{definition}[theorem]{Definition}
\newtheorem{example}[theorem]{Example}

\DeclareFontFamily{U}{UWCyr}{}
\DeclareFontShape{U}{UWCyr}{m}{n}{%
  <5> <6> <7> <8> <9>
  <10> <10.95> <12> <14.4> <17.28> <20.74> <24.88> wncyr10
  }{}
\DeclareMathAlphabet{\cyrm}{U}{UWCyr}{m}{n}
\DeclareSymbolFont{cyrm}{U}{UWCyr}{m}{n}
\DeclareSymbolFontAlphabet{\cyrm}{cyrm}
\DeclareMathSymbol{\Evo}{\cyrm}{cyrm}{"03}
\DeclareMathOperator{\byd}{\raisebox{-.2ex}{$\overset{\text{\tiny
def}}{=}$}}

\DeclareMathOperator{\sym}{sym}

\DeclareMathOperator{\rank}{rank}

\newcommand{\cprime}{\/{\mathsurround=0pt$'$}}

\newcommand{\R}{\mathbb{R}}

\newcommand{\cC}{\mathcal{C}}
\newcommand{\cE}{\mathcal{E}}

\newcommand{\cA}{\mathcal{A}}
\newcommand{\cI}{\mathcal{I}}
\newcommand{\cP}{\mathcal{P}}

\newcommand{\abs}[1]{\lvert#1\rvert}

\newcommand*{\pd}[2]{\mathchoice{\frac{\partial #1}{\partial #2}}
  {\partial #1/\partial #2}{\partial #1/\partial #2} {\partial%
  #1/\partial #2}}
\newcommand{\eval}[2][\right]{\relax
  \ifx#1\right\relax \left.\fi#2#1\rvert}

\let\abs=\envert

\providecommand{\href}[2]{#2} 

\allowdisplaybreaks[3]

\begin{document}
\title{Ordinary differential equations described by their Lie symmetry algebra\footnote{Work supported by GNFM, GNSAGA, Universit\`a di Messina,
    Universit\`a di Perugia, Universit\`a del Salento.}}

\author{Gianni Manno\thanks{INdAM -- Politecnico di Milano, Dipartimento di Matematica ``Francesco Brioschi'', via Bonardi 9, Milano, Italy,
    \texttt{giovanni.manno@polimi.it}} , Francesco
  Oliveri\thanks{Dipartimento di Matematica e Informatica, Universit\`a di
    Messina, Viale F. Stagno d'Alcontres 31, 98166 Messina, Italy,
    \texttt{francesco.oliveri@unime.it}} , Giuseppe
  Saccomandi\thanks{Dipartimento di Ingegneria Industriale, Universit\`{a}
    degli Studi di Perugia, 06125 Perugia, Italy,
    \texttt{saccomandi@mec.dii.unipg.it}} , Raffaele Vitolo\thanks{Dipartimento
    di Matematica ``E. De Giorgi'', Universit\`a del Salento, via per Arnesano,
    73100 Lecce, Italy, \texttt{raffaele.vitolo@unisalento.it}} }

\date{
\framebox{
\begin{minipage}{5.5cm}
\begin{center}
\small
Published in \emph{J. Geom.\ Phys.},
\\
 \textbf{85} (2014), 2--15.
\end{center}
\end{minipage}
}
}

\maketitle

\begin{abstract}
  The theory of Lie remarkable equations, \emph{i.e.}, differential equations
  characterized by their Lie point symmetries, is reviewed and applied to
  ordinary differential equations. In particular, we consider some relevant Lie
  algebras of vector fields on $\mathbb{R}^k$ and characterize Lie remarkable
  equations admitted by the considered Lie algebras.
\end{abstract}



\section{Introduction}

In the context of the geometric theory of symmetries of (systems of)
differential equations (DEs)
\cite{BlumanKumei:book,Many99,Ibragimov,Oli10,Olver1}, a natural problem is to
see when a DE, either partial (PDE) or ordinary (ODE), is uniquely determined
by its Lie algebra of point symmetries. The core of this paper is to
investigate the inverse problem in the context of ODEs: given a Lie algebra
$\mathfrak{s}$ of vector fields, how to construct ODEs having $\mathfrak{s}$ as
a Lie point symmetry subalgebra and satisfying some specific properties, which
will be clarified below, that ensure the uniqueness of such DE.  The idea of
describing DEs admitting a given Lie algebra of symmetries dates back at least
to S. Lie, who stated that $u_{xx}=0$ is the unique scalar $2^{\textrm{nd}}$ order
ODE, up to point transformations, admitting an $8$-dimensional Lie algebra of
symmetries.  Of course, a similar idea also applies to PDE: for instance, in
\cite{Rosenhausoneone} (see also \cite{Rosenhaus1,Rosenhaus2}), the author
proved that the only scalar $2^{\textrm{nd}}$ order PDE, with an unknown function and two
independent variables, admitting the Lie algebra of projective vector fields of
$\mathbb{R}^3$ as Lie point symmetry subalgebra is the Monge-Amp\`ere equation
$u_{xx}u_{yy}-u_{xy}^2=0$. The above idea plays a central role also in gauge
theories, where one wants to obtain information on differential operators
possessing a prescribed algebra of symmetries. The results of \cite{MPV} go in
this direction: $2^{\textrm{nd}}$ order field equations possessing translational and gauge
symmetries and the corresponding conservation laws (via Noether theorem) are
always derivable from a variational principle.

The standard procedure (also used in \cite{Rosenhausoneone}) for obtaining a
scalar DE admitting a prescribed Lie algebra of symmetries is that of computing
the differential invariants of its prolonged action, under some regularity
hypotheses; the invariant DE is then described by the vanishing of an arbitrary
function of such invariants. If the prolonged action is not regular, invariant
DEs can be obtained by a careful study of the singular set of the
aforementioned action. An efficient method for obtaining invariant scalar ODEs
in the latter case is that of using Lie determinants \cite{Olver2}, which we
shall employ for our purposes. See \cite{BlumanCole,Fu00,KruMor12,RWinter} for
more approaches to the problem.  In general, DEs do not possess a sufficient
number of independent Lie point symmetries able to characterize them (among the
others we recall KdV equation, Burgers' equation, Kepler's equations). In this
case, one can ask if they can be characterized by a more general algebra of
symmetries. A possible generalization of the concept of Lie remarkable
equations is that suggested in \cite{KruMor12,Rosenhausoneone}: this amounts to
extending the category of symmetries used in the definitions of Lie remarkable
equations to contact symmetries. For instance, the minimal surface equation of
$\mathbb{R}^3$ is completely determined by its contact symmetry algebra
\cite{Rosenhausoneone}. Also, an example of high-order Lie remarkable equation
in this `extended' sense is
\begin{equation}
  \label{eq:3}
  10u_{(3)}^3 u_{(7)} - 70u_{(3)}^2u_{(4)}u_{(6)} -
  49u_{(3)}^2u_{(5)}^2+ 280u_{(3)}u_{(4)}^2u_{(5)} - 175u_{(4)}^4= 0,
\end{equation}
where $u_{(k)}=d^ku/dx^k$, which possesses a $10$-dimensional Lie algebra of
contact symmetries (see \cite{DunSok11,Olver2}).  Sometimes, in order to
completely characterize a given DE, one should also consider non-local
symmetries. This is the situation discussed in \cite{Kra94}, where the idea of
complete symmetry group was proposed and exploited in order to characterize
uniquely Kepler's equation. This idea was subsequently exploited by several
authors in different ways for characterizing many differential equations
\cite{ADLT09,KL02,KL04,ALF01,Leach08,ML09,Nucci96}.

Following the terminology introduced in
\cite{MannoOliveriVitolo05,MOV07,MOV07b,Oli04,Oli06}, we call \emph{Lie
  remarkable} a DE which is completely characterized by its Lie algebra of
point symmetries.  Of course, this concept needs some cares and comments,
which we will give below. Thus, before giving a mathematical definition of it, we have to analyze all the requirements that can
make a DE unique, also by means of simple examples.
It is well known that, locally, any  $r^{\textrm{th}}$ order
differential equation $\cE$ with $n$ independent variables and $m$ dependent
ones can be interpreted as a submanifold of the $r$-jet $J^r(n,m)$ of the
trivial bundle $\mathbb{R}^n\times\mathbb{R}^m\to\mathbb{R}^n$. Let us denote
by $\sym(\cE)$ the Lie algebra of infinitesimal point symmetries of $\cE$. Thus, when saying that
$\cE$ is uniquely determined by $\sym(\cE)$, one should fix, as data of the
problem, the number of independent and dependent variables, the order of the DE
and its dimension as submanifold. For instance, (see also Section
\ref{sec.primitive}), the unique $5^{\textrm{th}}$ order ODE admitting the algebra of
projective vector fields of $\mathbb{R}^2$ is equation of item \ref{eq.conic}
of Theorem \ref{theo.main}, but also $u_{xx}=0$ is the unique $2^{\textrm{nd}}$ order ODE
admitting the algebra of projective vector fields of $\mathbb{R}^2$ as Lie
algebra of point symmetries. Remaining in the realm of projective algebra, the
system $\{y_{xx}=0, u_{xx}=0\}$ is uniquely determined by the $15$-dimensional
projective Lie algebra of $\mathbb{R}^3$, but, as we already said, also the
Monge-Amp\`ere equation $ u_{xx}u_{yy}-u_{xy}^2=0 $ admits the same
$15$-dimensional Lie algebra of vector fields as Lie algebra of point symmetries. Both the
system $\{y_{xx}=0, u_{xx}=0\}$ and $u_{xx}u_{yy}-u_{xy}^2=0$ are, in their own
class, the only DEs admitting the projective algebra of $\mathbb{R}^3$ as Lie
algebra of point symmetries.  As the last consideration, we observe that if an
equation $\cE$ admits a Lie algebra of point symmetries, also an open
submanifold of $\cE$ admits the same Lie algebra of symmetries, so that when
speaking about Lie remarkable equations one should think of them up to
inclusion.  Bringing all the above observations together, below we  formulate a
more precise definition of Lie remarkable equations.

\medskip\noindent
\textbf{\large Notations and conventions}: Throughout the paper, we will use the Einstein summation convention, unless otherwise specified. We will always use the word ``symmetry'' for ``infinitesimal point symmetry''. When we speak about a Lie algebra we always mean a \emph{Lie algebra of vector fields} of finite dimension, unless otherwise specified. Finally, if $\mathfrak{s}$ and $\mathfrak{g}$ are Lie algebras, $\mathfrak{s}\leq\mathfrak{g}$ means that $\mathfrak{s}$ is a Lie subalgebra of $\mathfrak{g}$.

\smallskip

\begin{definition}
  An $l$-dimensional $r^{\textrm{th}}$ order equation $\cE\subset J^r(n,m)$ is called \emph{Lie
    remarkable} if it is the only $l$-di\-men\-sional $r^{\textrm{th}}$ order equation in $J^r(n,m)$,
  up to inclusion and up to point transformations, admitting $\sym(\cE)$ as a Lie symmetry
  subalgebra.
\end{definition}

Below we will shed light on the above definition by means of a simple example.
Equation
\begin{equation*}
\mathcal{E}_1:\,\,u_{xx}=\frac{1}{2}u_x + e^{-2x}u_x^3
\end{equation*}
is not Lie-remarkable. In fact $\sym(\mathcal{E}_1)$ is linearly generated by
\begin{equation}\label{eq.symm.family}
  \partial_u\,, \quad \partial_x+u\partial_u\,, \quad
  u\partial_x+\frac{u^2}{2}\partial_u
\end{equation}
but also the equation
\begin{equation*}
\mathcal{E}_2:\,\,u_{xx}=\frac{1}{2}u_x
\end{equation*}
admits $\sym(\mathcal{E}_1)$ as a Lie subalgebra of its Lie symmetry algebra. Indeed, $\sym(\mathcal{E}_1)\lneq \sym(\mathcal{E}_2)$, as $\sym(\mathcal{E}_2)$ is isomorphic
to the projective Lie algebra of $\mathbb{R}^2$. Thus,
$\mathcal{E}_1$ and $\mathcal{E}_2$ are not equivalent. To conclude,
$\mathcal{E}_1$ is not Lie-remarkable, whereas $\mathcal{E}_2$ is.

\smallskip

Of course, an abstract Lie algebra can be realized, in terms of vector fields, in different non-equivalent ways. For instance, S. Lie \cite{Lie} investigated the possible realizations of the non-commutative Lie algebra of dimension $2$ (i.e., the Lie algebra spanned by two elements $X$ and $Y$ such that $[X,Y] = X$) as Lie algebra of vector
fields on $\mathbb{R}^2$. He showed that, almost every point of $\mathbb{R}^2$ has a neighborhood on which there are coordinates $x,u$ in which
\begin{equation}\label{eq.real}
1.\,\{X,Y\}=\{e^u\partial_u\,,\,-\partial_u\} \,\,\, \text{or} \,\,\, 2.\,\{X,Y\}=\{\partial_u\,,\,\partial_x+u\partial_u\}.
\end{equation}
Of course, realizations \eqref{eq.real} are not equivalent, as the orbits of the first realization are $1$-dimensional whereas the orbits of the second one are $2$-dimensional.

\begin{definition}\label{def.lie.rem}
We say that an $l$-dimensional $r^{\textrm{th}}$ order equation $\cE\subset J^r(n,m)$ is \emph{associated} with a Lie algebra of vector fields $\mathfrak{s}$ if it is the
  only $l$-di\-men\-sional $r^{\textrm{th}}$ order equation in $J^r(n,m)$,
  up to inclusion and up to point transformations, admitting $\mathfrak{s}$
  as a Lie subalgebra of $\sym(\cE)$.
\end{definition}
We would like to stress that, in Definition \ref{def.lie.rem}, choosing a realization of the Lie algebra in terms of vector fields is crucial, otherwise the definition of a DE associated with a Lie algebra would not be well posed. Indeed, non-equivalent realizations of the same (abstract) Lie algebra lead, in general, to different DEs. For instance, let us we consider the non-commutative Lie algebra of dimension $2$ and its realizations $1.$ and $2.$ of \eqref{eq.real}. The most general $2^{\mathrm{nd}}$ order ODE having $1.$ of \eqref{eq.real} as a Lie symmetry subalgebra is
\begin{equation}\label{eq.trivial}
u_{xx}=f(x)u_x+u_x^2
\end{equation}
whereas, if we consider $2.$ of \eqref{eq.real}, we obtain
\begin{equation}\label{eq.nontrivial}
u_{xx}=g\left(\frac{u_x}{u}\right)u.
\end{equation}
Finally we note that equations \eqref{eq.trivial} are point-equivalent, for any $f\in C^\infty(\mathbb{R})$, to $u_{xx}=0$ as they possess an $8$-dimensional Lie algebra of point symmetries (more directly, one can easily check that the Liouville-Cartan invariant vanishes), whereas equations of type \eqref{eq.nontrivial} are not point equivalent each other. In fact, equation \eqref{eq.nontrivial} with $g\left(\frac{u_x}{u}\right)=\left(\frac{u_x}{u}\right)^4$ cannot be linearizable as the only $2^{\mathrm{nd}}$ order ODEs which can have this property are of type $u_{xx}=h(x,u,u_x)$ with $h_{u_xu_xu_xu_x}=0$. Thus, we can say that equation $u_{xx}=0$ is associated with Lie algebra of vector fields $1.$ of \eqref{eq.real}, whereas there are no $2^{\textrm{nd}}$ order ODEs associated with Lie algebra of vector fields $2.$ of \eqref{eq.real}.

\smallskip

A first consequence of Definition \ref{def.lie.rem} is the following obvious proposition.
\begin{proposition}\label{prop.chissa}
  If the equation $\mathcal{E}$ is associated with $\mathfrak{s}\leq
  \sym(\mathcal{E})$, then it is associated with any subalgebra
  $\widetilde{\mathfrak{s}}$ of $\sym(\mathcal{E})$ such that
  $\mathfrak{s}\leq\widetilde{\mathfrak{s}}$.
\end{proposition}
We remark that DEs of different order can be associated to the same Lie
algebra of vector fields. From the above discussions it is clear that a Lie remarkable equation
needs a Lie algebra of point symmetries of suitable dimension: in Section
\ref{sec.first.int} we show that, in the case of scalar ODEs, this leads also
to the existence of first integrals.

\medskip

In the present paper we shall construct, in an algorithmic way, (system of)
ODEs associated with relevant Lie algebras of vector fields on $\mathbb{R}^k$
by using sufficient conditions contained in Section \ref{sec.comments} (more
precisely, Propositions \ref{prop.sufficienza}, \ref{crit} and
\ref{prop.sufficienza.2}). As first step, we obtain scalar Lie
remarkable ODEs by means of the local classification of primitive Lie algebras
of vector fields on $\mathbb{R}^2$ (a list of such Lie algebras of vector fields can be found in
\cite{Olver2}). Note that they include the euclidean, affine, special conformal
and projective Lie algebra of $\mathbb{R}^2$. Then we concentrate on the
computations of Lie remarkable systems of ODEs. Below we give the main
theorems.

\begin{theorem}\label{theo.main}
  Lie remarkable scalar ODEs associated with primitive Lie algebras of vector
  fields on $\mathbb{R}^2$ are listed below (we refer to
  table~\eqref{eq.primitive}) :
  \begin{enumerate}
  \item\label{eq.no.Lie.rem.scalar} There are no Lie remarkable scalar ODEs
    associated with Lie algebras \textbf{I}, \textbf{II} and \textbf{III};
  \item\label{eq.lines.2d} with algebras \textbf{IV} or \textbf{V} it is
    associated the equation of straight lines $u_{xx}=0\,$;
  \item\label{eq.GKS} with Lie algebra \textbf{VI} it is associated the
    equation $u_{xx}=0$ and the equation of the vanishing affine curvature
    $3u_{xx}u_{xxxx}-5u_{xxx}^2=0\,$;
  \item\label{eq.circle.2d} with Lie algebra \textbf{VII} it is associated the
    equation the equation of circles $(1+u_x^2)u_{xxx}-3{u_xu_{xx}^2}=0\,$;
  \item\label{eq.conic} with Lie algebra \textbf{VIII} it is associated the
    equation $u_{xx}=0$ and the equation of conic sections
    $9u_{xxxxx}u_{xx}^2+40u_{xxx}^3-45u_{xx}u_{xxx}u_{xxxx}=0\,$.
  \end{enumerate}
\end{theorem}
Equation of item \ref{eq.GKS} is also known as \emph{generalized
  Kummer-Schwartz equation} (see \cite{Leach08} for a discussion of this
topic). Equation of item \ref{eq.circle.2d} can be realized as the vanishing of
total derivative of the euclidean curvature $u_{xx}(1+u_x^2)^{-\frac{3}{2}}$ of
the curve $u=u(x)$. As regard to equation of item \ref{eq.conic}, it was
somehow expectable to obtain it: indeed \textbf{VIII} is the projective Lie
algebra of $\mathbb{R}^2$ and a projective transformation sends a conic curve
into a conic curve.

\smallskip

For what concerns systems of ODEs, by means of the methods
described above, we found Lie remarkable systems of ODEs in $2$ dependent
variables associated with euclidean, affine, conformal and projective Lie
algebra of $\mathbb{R}^3$. We summarize our results in the theorem below.

\begin{theorem}\label{thm:systems}
  Lie remarkable systems of ODEs in $2$ dependent variables associated with
  isometry, affine, special conformal and projective Lie algebra of the
  euclidean space $\mathbb{R}^{3}$ are, respectively, listed below
  \begin{enumerate}
  \item\label{eq.lines.nd} with the isometry Lie algebra it is associated the
    system of straight lines $\{u^k_{xx}=0,\,\,k=1,2\}\,$;
  \item\label{eq.affine} with the affine Lie algebra it is associated the system
    $\{u^k_{xx}=0,\,\,k=1,2\}\,$ and the system appearing in section
    \ref{sec.aff.as.Lie.sym};
  \item\label{eq.circle.nd} with the conformal Lie algebra it is associated the
    system of circles $
    \{(1+\sum_{j}(u^j_x)^2)u^k_{xxx}=3u^k_{xx}{\sum_{j}u^j_xu^j_{xx}}\,,\,\,
    k=1,2\,\}$;
  \item with the projective Lie algebra it is associated the system
    $\{u^k_{xx}=0,\,\,k=1,2\}\,$ and the system appearing in section
    \ref{sec.proj.as.Lie.sym}.
  \end{enumerate}
\end{theorem}

We see that the equation/system of lines in the euclidean space appears many
times in the Theorems above in view of Proposition \ref{prop.chissa}. We
underline that it is known that a system of ODEs possessing a Lie symmetry
algebra of maximal dimension (i.e. $m^2+4m+3$) is point-equivalent to the
system of lines \cite{Fel}. In this respect, the result of item
\ref{eq.lines.nd} of Theorem \ref{thm:systems} is somehow unexpected: note
that, in the scalar case, a $2^{\textrm{nd}}$ order ODE which admits the $3$-dimensional
Lie algebra of infinitesimal isometries of the euclidean space as Lie symmetry
subalgebra is not necessarily the equation $u_{xx}=0$ (see the discussion
contained in Section \ref{sec.iso.as.Lie.symm}).
Anyway in Section \ref{sec.iso.as.Lie.symm} we prove the result of item \ref{eq.lines.nd} for an arbitrary number of dependent variables. We observe that the result of item \ref{eq.circle.nd} holds also
for a number of dependent variables less or equal than four. The higher-order case appearing
  in item \ref{eq.affine} is discussed more in detail in Section \ref{sec.aff.as.Lie.sym}.

\medskip

All computations are performed through the use of the
computer algebra package ReLie \cite{Relie}, a REDUCE program developed by one
of us (F.O.).

\section{Preliminaries}
\label{sec.preliminary.definitions}

In the whole paper, all manifolds and maps are supposed to be $C^\infty$.  Here
we recall some basic facts regarding jet spaces (for more details, see
\cite{Many99,Olver1}).  In what follows, $\lambda$ and $\mu$ run from $1$ to
$n$ whereas $i$ and $j$ run from $1$ to $m$. By $J^r(n,m)$ we denote the
$r^{\textrm{th}}$ order jet space of the trivial projection
$\mathbb{R}^n\times\mathbb{R}^m\to\mathbb{R}^n$. Note that
$J^0(n,m)=\mathbb{R}^n\times\mathbb{R}^m$. A system of coordinates
$(x^\lambda,u^i)$ on $\mathbb{R}^n\times\mathbb{R}^m$ induces a system of
coordinates $(x^\lambda,u^i_{\sigma})$ on $J^r(n,m)$, where
$\sigma=(\sigma_1,\sigma_2,\dots,\sigma_n)\in\mathbb{N}_0^n$ such that $|\sigma|:=\sum_i\sigma_i\leq r$, in
the following way. For each local map
$s\colon I\subset\mathbb{R}^n\to \mathbb{R}^m$, we define its jet prolongation
$j_rs\colon I\to J^r(n,m)$ in such a way that
\begin{displaymath}
  u^i_{\sigma} \circ j_rs = \frac{\partial^{|\sigma|}(u^i\circ s)}
  {(\partial x^1)^{\sigma_1}\cdots(\partial x^n)^{\sigma_n}}\,,\,\,|\sigma|\leq r.
\end{displaymath}

In the case $n=1$, i.e. one independent variable, we also denote $x^1$ by $x$ and $u^i_\sigma=u^i_{(\sigma_1)}$ both by $u^i_{\sigma_1}$ and
$u^i_{\underset{\sigma_1-times}{\underbrace{xx\dots x}}}$.

On $J^r(n,m)$ there is the (higher) contact distribution which is generated by
the vector fields
\begin{equation*}
  D_{\lambda }\byd\frac{\partial }{ \partial x^{\lambda }}+
  u_{\sigma,\lambda }^{j}
  \frac{\partial }{\partial u_{\sigma }^{j}}
  \quad\text{and}\quad \pd{}{u_{\tau}^{j}},
\end{equation*}
where $0\leq\abs{\sigma}\leq r-1$, $\abs{\tau}=r$ and $\sigma,\lambda$ denotes
the multi-index $(\sigma_1,\ldots,\sigma_\lambda +1,\ldots,\sigma_n)$. We note
that the contact distribution is spanned by tangent vectors to all submanifolds
of $J^r(n,m)$ of the type $j_rs(I)$; conversely, an integral $n$-dimensional
manifold of the contact distribution which projects surjectively on $I$ under
the canonical map $(x^\lambda,u^i_\sigma)\mapsto (x^\lambda)$ is locally of the
form $j_rs(I)$.

Any vector field $X$ on $J^0(n,m)$ can be lifted to a vector field $X^{(r)}$ on
$J^r(n,m)$ by lifting its local flow: such vector field preserves the contact
distribution.  In coordinates, if $X=X^\lambda\pd{}{x^\lambda}+X^i\pd{}{u^i}$
is a vector field on $J^0(n,m)$, then its $r$-lift $X^{(r)}$ has the following
coordinate expression
\begin{equation}
  \label{eq:1}
  X^{(r)} = X^\lambda\pd{}{x^\lambda} + X^i_{\sigma}\pd{}{u^i_{\sigma}},
\end{equation}
whose components are iteratively defined by $X^j_{\tau,\lambda} =
D_{\lambda}(X^j_{\tau})-u^j_{\tau,\mu}D_{\lambda}(X^{\mu})$ with $|\tau|<r$. In the case $n=1$, i.e. one independent variable, we also denote $X^i_\sigma=X^i_{(\sigma_1)}$ by $X^i_{\sigma_1}$.

An \emph{$r^{\textrm{th}}$ order differential equation (DE) $\mathcal{E}$ with $n$
  independent variables and $m$ unknown functions (or dependent variables)} is
a submanifold of $J^r(n,m)$.

A solution is an $n$-dimensional submanifold of $J^0(n,m)$ which projects
surjectively on $\mathbb{R}^n$ and such that its $r$-prolongation is contained
in $\mathcal{E}$.  An \emph{infinitesimal point symmetry} of $\cE$ is a vector
field $X$ on $J^0(n,m)$ such that its $r$-prolongation $X^{(r)}$ is tangent to
$\cE$: they transform solutions into solutions. We denote by $\sym(\cE)$ the
Lie algebra of infinitesimal point symmetries of the equation $\cE$.

Let $\cE$ be locally described by $\{F^i=0\}$, $i=1\dots k$ with $k<\dim
J^r(n,m)$. Then finding point symmetries amounts to solve the system
\begin{equation*}
  X^{(r)}\left(F^i\right) =0 \quad\text{whenever}\quad F^i=0.
\end{equation*}
The problem of determining the Lie algebra $\sym(\cE)$ is said to be the
\emph{direct Lie problem}. Conversely, given a Lie subalgebra $\mathfrak{s}$ of
the Lie algebra of the vector fields on $J^0(n,m)$, we consider the
\emph{inverse Lie problem}, \emph{i.e.}, the problem of characterizing
equations $\cE\subset J^r(E,n)$ such that $\mathfrak{s}\subset \sym(\cE)$.

\smallskip In the present paper we mainly deal with (system of)
ordinary differential equations (ODEs), and we assume that they always can be
put in normal forms. Thus, by definition, \emph{an $r^{\textrm{th}}$ order ODE is the image of a section
  of the bundle $J^r(1,m)\to J^{r-1}(1,m)$}. Let $\mathfrak{s}$ be a Lie algebra of vector fields on $J^0(1,m)$. Let $\{X_a\}_{1\leq a \leq k}$ be a basis of $\mathfrak{s}$.
We denote by $\mathcal{M}_{\mathfrak{s}^{(r)}}$ the $k\times (1+mr+m)$  matrix of the components, w.r.t. the basis $\{\partial_{x^1},\partial_{u^i_{\sigma_1}}\}_{\substack{1\leq i\leq m\\ 0\leq\sigma_1\leq r}}$, of the prolongations $X_k^{(r)}$ to $J^r(1,m)$ of each
  $X_k$. Namely, if, according with~\eqref{eq:1}, $X_k^{(r)} =
  X_k{}^1\pd{}{x^1} + X_k{}^i_{\sigma}\pd{}{u^i_\sigma}$,
 then matrix $\mathcal{M}_{\mathfrak{s}^{(r)}}$ is
\begin{equation}\label{eq.matrice.prolungamenti}
\mathcal{M}_{\mathfrak{s}^{(r)}}=
\left(
\begin{array}{ccccccccccc}
X_1{}^1 & X_1{}^1_0 & \cdots\cdots & X_1{}^m_0 & X_1{}^1_1 & \cdots\cdots & X_1{}^m_1 & \cdots\cdots & X_1{}^1_r & \cdots\cdots & X_1{}^m_r
\\
\vdots  & \vdots & \vdots & \vdots & \vdots & \vdots & \vdots & \vdots & \vdots &  \vdots & \vdots
\\
\vdots  & \vdots & \vdots & \vdots & \vdots & \vdots & \vdots & \vdots & \vdots &  \vdots & \vdots
\\
X_k{}^1 & X_k{}^1_0 & \cdots\cdots & X_k{}^m_0 & X_k{}^1_1 & \cdots\cdots & X_k{}^m_1 & \cdots\cdots & X_k{}^1_r & \cdots\cdots & X_k{}^m_r
\end{array}
\right)
\end{equation}

We omit the dependency of the above matrix on the basis of $\mathfrak{s}$ as
the computations we have performed and which made use of matrix
\eqref{eq.matrice.prolungamenti} are independent of the chosen basis. In fact
we shall mainly deal with the rank of \eqref{eq.matrice.prolungamenti}.

\section{Sufficient conditions for Lie remarkability and relationship with
  first integrals}\label{sec.comments}

\subsection{Sufficient conditions we shall use for constructing Lie remarkable
  ODEs}

The definition of Lie remarkable equations leads naturally to some sufficient
conditions for determining them. To start with, we observe that with any Lie
algebra $\mathfrak{s}$ of vector fields on a manifold $M$ it is associated an
involutive distribution $\mathcal{D}^\mathfrak{s}$ (generally, of non-constant
rank) defined by
\begin{equation}\label{eq.distr.sing}
p\in M\mapsto\mathcal{D}^\mathfrak{s}_p:=\{X_p\,\,|\,\,
X\in\mathfrak{s}\}\subset T_pM.
\end{equation}
In view of Frobenius theorem, involutive distributions on a smooth manifold $M$ are integrable, i.e. through each point of $M$ there is a unique maximal leaf, provided they are of constant rank.
For involutive distributions of non-constant rank, there exist sufficient conditions which assure their integrability (see for instance, Theorem 3.25 of \cite{KMS}).
For distributions coming from Lie algebra actions, i.e. of type \eqref{eq.distr.sing}, it holds the following theorem.
\begin{theorem}[\cite{AlekMich}]\label{th.Alek.Michor}
Let $\mathfrak{s}$ be a finite-dimensional Lie algebra. Then distribution $\mathcal{D}^\mathfrak{s}$ defined by \eqref{eq.distr.sing} is integrable.
\end{theorem}
We note that, with any Lie symmetry algebra $\sym(\mathcal{E})$ of a
differential equation $\mathcal{E}\subset J^r(n,m)$ of order $r$, one can
associate the distribution $\mathcal{D}^{\sym(\mathcal{E})}$ on the $r^{\textrm{th}}$ order
jet space. The following inequality holds:
\begin{equation*}
  \quad \dim\sym(\cE) \geq \dim\mathcal{D}^{\sym(\mathcal{E})}_\theta, \,\, \forall\,\theta\in
  J^r(n,m),
\end{equation*}
where $\dim\sym(\cE)$ is the dimension, as real vector space, of the Lie
algebra $\sym(\cE)$ of the infinitesimal point symmetries of $\cE$.  An integral
submanifold of $\mathcal{D}^{\sym(\mathcal{E})}$ is, in general, an equation in
$J^r(n,m)$. By construction, such equation admits all elements in $\sym(\cE)$
as infinitesimal point symmetries. This leads to the following proposition,
which we shall use in Section \ref{sec.proof} for computing Lie remarkable
equations starting from distinguished Lie algebras of vector fields on
$J^0(n,m)$.

\begin{proposition}\label{prop.sufficienza}
  Let $\cE\subset J^r(n,m)$. If $\dim \mathcal{D}^{\sym(\mathcal{E})}_\theta >
  \dim\cE$ $\forall\,\theta\in J^r(n,m)\setminus\{\mathcal{E}\cup \mathcal{F}\}$, where $\mathcal{F}$ is either an empty set or a finite union of submanifolds of dimension less than $\dim\mathcal{E}$, then $\cE$ is Lie remarkable.
\end{proposition}
\begin{remark}
The hypotheses of Proposition \ref{prop.sufficienza} are motivated by Theorem \ref{th.Alek.Michor}, which we shed light by an example. Let us consider the Lie algebra of vector fields $\mathfrak{s}$ on $\mathbb{R}^3=(x,y,z)$ linearly generated by
$$
x\partial_x+y\partial_y\,, \quad -y\partial_x+x\partial_y\,, \quad z\partial_z
$$
In this case, distribution $\mathcal{D}^{\mathfrak{s}}$ is of non-constant rank: indeed it has rank $3$ on $\mathbb{R}^3$ minus the algebraic variety described by $z(x^2+y^2)=0$. The hyperplane $S$ described by $z=0$ is the unique $2$-dimensional submanifold such that $\mathcal{D}^{\mathfrak{s}}_p\subseteq T_pS$, $\forall\,p\in S$. Note that, outside $S$, the rank of $\mathcal{D}^{\mathfrak{s}}$ is equal to $3$ except on the two lines $\{x=0,y=0,z>0\}$ and $\{x=0,y=0,z<0\}$, where the rank is equal to $1$.
\end{remark}
\begin{proof}[Proof of Proposition \ref{prop.sufficienza}]
Let $\widetilde{\mathcal E}\subset J^r(n,m)$ be a DE (of the same dimension as $\mathcal{E}$) such that $\sym({\mathcal E})\leq\sym(\widetilde{\mathcal E})$. Let us suppose that $\widetilde{\mathcal E}\neq {\mathcal E}$, so that there exists at least a point of $\widetilde{\mathcal E}$ which does not belong to $\mathcal{E}$. Let us denote such a point by $\theta$. As a first case, let us assume that $\theta\in\widetilde{\mathcal E}\setminus \{\mathcal{E}\cup\mathcal{F}\}$. Since $\sym({\mathcal E})$ is a Lie subalgebra of the Lie symmetry algebra of $\widetilde{\mathcal{E}}$, then $\mathcal{D}^{\sym(\mathcal{E})}_\theta\subseteq T_\theta\widetilde{\mathcal{E}}$. This inclusion implies that $\dim(\mathcal{D}^{\sym(\mathcal{E})}_\theta)\leq\dim(\widetilde{\mathcal E})=\dim(\mathcal{E})$, which contradicts the hypothesis.

Let us now assume that $\theta\in (\widetilde{\mathcal E}\setminus\mathcal{E})\cap\mathcal{F}$. This implies that $\widetilde{\mathcal E}\cap\mathcal{F}\neq 0$. For dimensional reasons, $\widetilde{\mathcal E}\neq\mathcal{F}$, so that there exists a point in $\widetilde{\mathcal E}\setminus\mathcal{F}$ which, in its turn, does not belong to $\mathcal{E}$. Then it is enough to apply the reasoning of the previous case.
\end{proof}
From now on we shall concentrate only on ODEs since they are the target of our
investigation.

\smallskip

We would like to underline that, if an ODE $\cE\subset J^r(1,m)$ satisfies
Proposition \ref{prop.sufficienza} and it can be put in normal form (i.e. $\cE$
is the image of a section of the bundle $J^r(1,m)\to J^{r-1}(1,m)$, in
particular, it is a determined system), then also the $k$-prolongation
$\cE^{(k)}$ of $\cE$, with $k\in\mathbb{N}$ arbitrary, is Lie
remarkable. Indeed, in this case, $\dim\cE^{(k)}=\dim\cE$ and the rank of the
distribution spanned by the symmetries cannot decrease.

Proposition \ref{prop.sufficienza} suggests that, in order to construct an
$l$-dimensional Lie remarkable equation, one can start from an
$(l+1)$-dimensional Lie algebra of vector fields on
$J^0(n,m)=\mathbb{R}^{n+m}$. This leads, in the case of computation of scalar
Lie remarkable ODEs of $r^{\textrm{th}}$ order, for dimensional reasons, to consider the
\emph{Lie determinant} associated to an $(r+2)$-dimensional Lie algebra of
vector fields (see the beginning of Section \ref{sec.proof}).  Actually, in
some cases, an $l$-dimensional ODE can be uniquely determined by an
$l$-dimensional Lie algebra, as the content of the next section shows.

\subsection{Lie remarkable ODEs determined by a lower dimensional
  Lie algebra of vector fields: foliations of equivalent ODEs and
  (pseudo)-stabilization order}

We have already seen that Lie algebra of vector fields $1.$ of \eqref{eq.real} is sufficient to completely characterize the ODE $u_{xx}=0$, up to point transformations (see the discussion after Definition \ref{def.lie.rem}). In this section we investigate other situations in which it is possible to construct Lie remarkable
ODEs starting from a Lie algebra of vector fields of lower dimension w.r.t. that of Proposition
\ref{prop.sufficienza}.

\smallskip

 Consider
realization \eqref{eq.symm.family} (in terms of vector fields on $\mathbb{R}^2$) of the Lie algebra $\mathfrak{sl}(2,\mathbb{R})$.
Second order ODEs admitting vector fields \eqref{eq.symm.family} as Lie point
symmetries are
\begin{equation}\label{eq.family}
  u_{xx}=\frac{1}{2}u_x + Ke^{-2x}u_x^3\,, \quad K\in\mathbb{R}.
\end{equation}
Equation \eqref{eq.family} and vector fields \eqref{eq.symm.family} appeared in \cite{BMM} in the context of the local classification of projective structures on a $2$-dimensional manifold.
A priori, DEs belonging to the above $1$-parametric family could be all
point-equivalent, so that one could consider the equation \eqref{eq.family}
with $K=0$ as its representative, that is the only equation (up to point
transformations) admitting the above Lie algebra of vector fields as a
subalgebra of Lie point symmetries. A deeper study shows that it is not the
case. Indeed, even if the change $x\to x+c$, for some suitable constant $c$,
allows to say that all equations \eqref{eq.family} with $K>0$ (respectively,
$K<0$) are point equivalent, for $K=0$ equation \eqref{eq.family} admits the
$8$-dimensional Lie algebra $\mathfrak{sl}(3,\mathbb{R})$ as Lie algebra of point
symmetries, so that it is not point-equivalent with any equations
\eqref{eq.family} with $K\neq 0$. We stress that the sufficient criterion given
in Proposition \ref{prop.sufficienza} is not fulfilled for $K\neq 0$, the above
equation being a $3$-dimensional submanifold of the $4$-dimensional jet space
$J^2(1,1)$. On the other hand, for $K=0$, the criterion is fulfilled if we
consider the Lie algebra $\mathfrak{sl}(3,\mathbb{R})$ and the equation
$u_{xx}=\frac{1}{2}u_x$ is Lie remarkable (note that it is equivalent to
$u_{xx}=0$ as it admits a $8$-dimensional Lie symmetry algebra). From a
theoretical viewpoint, obtaining a foliation of point-equivalent equations is
possible. For instance, the most general scalar $2^{\textrm{nd}}$ order PDE admitting the
Lie algebra $\mathfrak{s}$ linearly generated by $\{\partial_x,\partial_u,x\partial_u\}$ as a Lie symmetry subalgebra is $u_{xx}=K$, $K\in\mathbb{R}$, as $u_{xx}$ is the
only $2^{\textrm{nd}}$ order differential invariant of $\mathfrak{s}$ (up to functional
dependence).  All these equations are point equivalent to $u_{xx}=0$, as they
admit the $8$-dimensional projective Lie algebra $\mathfrak{sl}(3,\mathbb{R})$
as Lie point symmetry algebra.  Then we can say that the Lie algebra of vector fields
$\mathfrak{s}$ uniquely determines equation $u_{xx}=0$ up to point
transformations, i.e. it is associated with $\mathfrak{s}$ in the
sense of Definition \ref{def.lie.rem}.  Thus, we have the following proposition:

\begin{proposition}\label{crit}
  Let $\mathfrak{s}$ be a Lie algebra of vector fields on $J^0(1,1)$. Let us
  suppose that there exists some $r\in \mathbb{N}$ such that the rank of
  $\mathcal{D}^{\mathfrak{s}^{(r)}}$ is of codimension $1$ almost
  everywhere. Let $F$ be the unique function on $J^r(1,1)$ (up to functional
  dependence) such that $\mathcal{D}^{\mathfrak{s}^{(r)}}(F)=0$. Then if the
  equations $F=k$, $k\in\mathbb{R}$ are point equivalent each other, then
  equation $F=0$ is associated with $\mathfrak{s}$.
\end{proposition}
\begin{proof}
Since $F$ is, by hypothesis, the only differential invariant of $r^\textrm{th}$ order of the Lie algebra $\mathfrak{s}$, $F=k$, with $k\in\mathbb{R}$, are all and only the $r^\textrm{th}$ order DEs admitting $\mathfrak{s}$ as a Lie symmetry subalgebra. Now, if all these DEs are equivalent each other, they will be equivalent, in particular, to equation $F=0$.
\end{proof}

Below we explain another way of obtaining Lie remarkable
ODEs starting from a Lie algebra of vector fields of lower dimension w.r.t. that of Proposition
\ref{prop.sufficienza}. This is based on the fact that some $r^{\textrm{th}}$ order ODE
can be uniquely constructed by a foliation of $(r-1)^{\textrm{th}}$ order ODEs.  Below we
shall be more precise by considering a concrete example.  The equation
$\cE:=\{u_{xx}=0\}\subset J^2(1,1)$ can be uniquely constructed starting from
the foliation of $J^1(1,1)$ given by $\mathcal{F}_k=0$, where
$\mathcal{F}_k:=u_x-k$, $k\in\mathbb{R}$. Indeed, we have that $\cE=\bigcup_k
\mathcal{F}_k^{(1)}$.  Even though we can construct the equation $u_{xx}=0$
starting from such a foliation of $J^1(1,1)$, the two Lie point symmetries
$\partial_x$ and $\partial_u$ of $u_x=k$ are not enough to determine the
equation $u_{xx}=0$ in the Lie remarkable sense (i.e. there are many $2^{\textrm{nd}}$
order ODEs associated with the Lie algebra of translations) although they
determine uniquely the foliation $\mathcal{F}_k=0$ of $J^1(1,1)$.  To
completely determine equation $u_{xx}=0$ in the Lie remarkable sense, the
foliation $u_x=k$ given by a Lie algebra of vector fields should live in
$J^2(1,1)$ rather than in $J^1(1,1)$ (then, in this case, we need $3$
symmetries rather than $2$). Indeed, such a foliation and equation $u_{xx}=0$
represent, essentially, the same object in $J^2(1,1)$.  This trivial
observation leads to an interesting consequence. If we find a $3$-dimensional
Lie algebra of vector fields such that its orbits in $J^2(1,1)$ form the
foliation $u_x=k$, then such algebra completely determines the equation
$u_{xx}=0$. This Lie algebra of vector fields exists: indeed, the $3$ infinitesimal homotheties
of $\mathbb{R}^2$ (two translations and the stretching) are Lie point
symmetries of $u_x=k$ for any $k$. This $3$-dimensional Lie algebra is enough
to completely determine $u_{xx}=0$ in the Lie remarkable sense. So, we
constructed a Lie remarkable equation of dimension $3$ starting from a
$3$-dimensional Lie algebra. In more precise words, we have the following
proposition.

\begin{proposition}\label{prop.sufficienza.2}
  Let $\mathfrak{s}$ be a Lie algebra of vector fields on $J^0(1,1)$. Let us
  suppose that there exists an integer $r$ such that the system
  $\mathfrak{s}^{(r)}(F)=0$ has a unique solution $F$ (up to functional
  dependence).  If $F\in C^\infty(J^{r-1}(1,1))$ but $F\notin
  C^\infty(J^{r-2}(1,1))$, then equation $D_x(F)=0$ is Lie remarkable and it is
  associated with $\mathfrak{s}$.
\end{proposition}
\begin{proof}
  Since $\mathfrak{s}^{(r)}(F)=0$, we have that $\mathfrak{s}\subset
  \sym(\{F=k\})\,\,\forall\,k\in\mathbb{R}$, that implies
  $\mathfrak{s}\subset\sym(\{D_x(F)=0\})$. Thus, on one hand
  $\mathfrak{s}^{(r)}$ spans almost everywhere a codimension $1$ distribution
  on $J^r(1,1)$ and, on the other hand, it is tangent to the hypersurface
  $\{D_x(F)=0\}$. Since the only solution to $\mathfrak{s}^{(r)}(F)=0$ is an
  $F$ of order $(r-1)$ but not $(r-2)$, $D_x(F)$ is of order $r$ (but not
  $(r-1)$), so that the only possibility for $\mathfrak{s}$ to be a subalgebra
  of Lie point symmetries of $\{D_x(F)=0\}$ is that such equation is contained
  in the set
$$
\mathrm{sing}(\mathfrak{s}^{(r)}):=\{\text{points of $J^r(1,1)$ where the rank
  of $\mathfrak{s}^{(r)} $is not maximal}\}.
$$
By hypothesis, the rank of $\mathfrak{s}^{(r)}$ is maximal almost everywhere
and it is equal to $\dim J^r(1,1)-1=r+1$.  Below we see that the set
$\mathrm{sing}(\mathfrak{s}^{(r)})$ contains only a hypersurface of order $r$
which is decribed by $\{D_x(F)=0\}$. This comes from the fact that
$\mathrm{sing}(\mathfrak{s}^{(r)})$ is given by a system $\{D_x(F)G_i=0\}$,
where $G_i$ are smooth functions on $J^{r-1}(1,1)$. In fact, the set
$\mathrm{sing}(\mathfrak{s}^{(r)})$ is described by the system formed by all
determinants of $(r+1)\times (r+1)$ submatrices of
$\mathcal{M}_{\mathfrak{s}^{(r)}}$ (see \eqref{eq.matrice.prolungamenti} for
the definition) equating to zero. Since in the last column of
$\mathcal{M}_{\mathfrak{s}^{(r)}}$ the highest order derivatives appears at
first degree, determinants of submatrices containing elements of the last
column are polynomial of first degree in the highest order derivatives. Since
$D_x(F)$ is exactly of order $r$ (the highest), functions $G_i$ cannot be of
highest order. We conclude that equation $\{D_x(F)=0\}$ is the only
hypersurface of order $r$ contained in $\mathrm{sing}(\mathfrak{s}^{(r)})$.

In view of the above reasonings, any other $r^{\textrm{th}}$ order scalar ODE $\mathcal{E}$
admitting $\mathfrak{s}$ as a Lie symmetry subalgebra is such that
$\mathcal{E}\subset \mathrm{sing}(\mathfrak{s}^{(r)})$.  Being $\mathcal{E}$ a
hypersurface of order $r$, in view of the above conclusion we have
$\mathcal{E}=\{D_x(F)=0\}$.
\end{proof}

\begin{remark}\label{rem.nec.cond.suff.2}
  To satisfy the hypotheses of Proposition \ref{prop.sufficienza.2}, a
  necessary condition is that, almost everywhere,
  $\rank(\mathcal{M}_{\mathfrak{s}^{(r-1)}})=r$ and
  $\rank(\mathcal{M}_{\mathfrak{s}^{(r)}})=r+1$. Lie algebras which
  \emph{pseudo-stabilizes} in the sense specified in \cite{Olver2} satisfies
  such condition.
\end{remark}

\begin{remark}\label{rem.controesempio}
  The fact that $\mathfrak{s}^{(r)}(F)=0$ has a unique solution which, in its
  turn, is exactly of $(r-1)^{\textrm{th}}$ order, implies also that
  $\mathfrak{s}^{(r-1)}(F)=0$. One can ask if we can weaker such hypothesis by
  assuming only that $\mathfrak{s}^{(r-1)}(F)=0$ has a unique solution of
  $(r-1)^{\textrm{th}}$ order. In this case there are examples showing that $D_x(F)=0$ is not
  Lie remarkable. We construct such an example immediately after
  \eqref{eq.diff.inv.1}.
\end{remark}

\begin{example}
  Let us consider the Lie algebra $\mathfrak{s}$ linearly generated by the following
  vector fields on $J^0(1,1)$:
$$
\partial_x, \quad \partial_u, \quad x\partial_u, \quad x\partial_x
+2u\partial_u\,.
$$
We have that
$$
\mathcal{M}_{s^{(3)}}=\left(
  \begin{array}{ccccc}
    1 & 0 & 0 & 0 & 0\\
    0 & 1 & 0 & 0 & 0\\
    0 & x & 1 & 0 & 0\\
    x & 2u & u_x & 0 & -u_{xxx}
  \end{array}
\right)\,.
$$
The above matrix has rank equal to $4$ outside the hypersurface $u_{xxx}=0$. We
also note that the only solution to system $\mathfrak{s}^{(3)}(F)$ is
$F=F(u_{xx})$. Thus, in view of Proposition \ref{prop.sufficienza.2}, we obtain
that $u_{xxx}=0$ is the only $3^{\textrm{rd}}$ order ODE associated with $\mathfrak{s}$.
\end{example}

The phenomenon described above does not appear in the case of PDEs as, in
general, one needs several foliations of $1^{\textrm{st}}$ order PDEs to reconstruct a
$2^{\textrm{nd}}$ order PDE. In fact, the above construction, for PDEs, corresponds to the
existence of intermediate integrals, as, by definition, an intermediate
integral of an $r^{\textrm{th}}$ order PDE $\cE$ is a function $f$ on the jet space of order
$r-1$ such that all solutions of the family $f=c$, $c\in\mathbb{R}$ are also
solutions of $\cE$. For instance, if a Monge-Amp\`ere equation with two
independent variables admits two special intermediate integrals, then it can be
reconstructed starting from them (see for instance \cite{AMP_Duke}).

\subsection{Lie remarkability and existence of first
  integrals}\label{sec.first.int}

In view of all that we said so far, it is expectable that Lie remarkable DEs
must have a suitable number of symmetries in order to be uniquely
determined. This implies, in the case of scalar ODEs, also the existence of
first integrals, as the following proposition shows.

\begin{proposition}\label{prop.simm.first}
  Let us consider the following ODE
  \begin{equation}\label{eq.ODE.n}
    u_{(n)}=f(x,u,u_{(1)},\dots,u_{(n-1)})\,, \quad\text{where}\,\,u_{(k)}:=\frac{d^ku}{dx^k} 
  \end{equation}
  Let $X_1,\dots,X_m$, $m\geq n$, be (point) symmetries of
  \eqref{eq.ODE.n}. Then,
  \begin{equation}\label{eq.first.int}
    \frac{X_{i_1}\lrcorner X_{i_2} \lrcorner \cdots \lrcorner X_{i_n}
      \lrcorner Z \lrcorner \Omega}{X_{j_1}\lrcorner X_{j_2} \lrcorner
      \cdots \lrcorner X_{j_n}\lrcorner Z \lrcorner \Omega},
  \end{equation}
  where
$$
Z:=\partial_x + u_{(1)}\partial_{u} + \cdots + u_{(n-1)}\partial_{u_{(n-2)}} +
f\partial_{u_{(n-1)}}\,, \,\,\Omega=dx\wedge du \wedge du_{(1)}\wedge \cdots
\wedge du_{(n-1)},
$$
is a first integral of \eqref{eq.ODE.n}.
\end{proposition}
\begin{proof}
  We observe that
  \begin{equation*}
    L_Z\left( \frac{1}{X_{i_1}\lrcorner X_{i_2} \lrcorner \cdots
        \lrcorner X_{i_n}\lrcorner Z \lrcorner \Omega}\,\Omega \right)=0.
  \end{equation*}
  Indeed,
  \begin{multline*}
    L_Z\left( \frac{1}{X_{i_1}\lrcorner X_{i_2} \lrcorner \cdots \lrcorner
        X_{i_n}\lrcorner Z \lrcorner \Omega}\,\Omega \right)
    =-\sum_k \frac{X_{i_1}\lrcorner \cdots \lrcorner
      X_{i_{k-1}}\lrcorner[Z,X_{i_k}]\lrcorner X_{i_{k+1}}\lrcorner \cdots
      \lrcorner Z\lrcorner\Omega}{(X_{i_1}\lrcorner X_{i_2} \lrcorner \cdots
      \lrcorner X_{i_n}\lrcorner Z \lrcorner \Omega)^2}\Omega
    \\
    -\frac{X_{i_1}\lrcorner X_{i_2} \lrcorner \cdots \lrcorner X_{i_n}\lrcorner
      Z \lrcorner L_Z\Omega}{(X_{i_1}\lrcorner X_{i_2} \lrcorner \cdots
      \lrcorner X_{i_n}\lrcorner Z \lrcorner \Omega)^2} \Omega +
    \frac{1}{X_{i_1}\lrcorner X_{i_2} \lrcorner \cdots \lrcorner
      X_{i_n}\lrcorner Z \lrcorner \Omega}L_Z\Omega =0.
  \end{multline*}
  We obtained the last equality in view of the following facts. Since $X_{i_k}$
  are symmetries of \eqref{eq.ODE.n}, $[Z,X_{i_k}]$ is proportional to $Z$ for
  any $1\leq i_k\leq m$, so that
$$
X_{i_1}\lrcorner \cdots \lrcorner X_{i_{k-1}}\lrcorner[Z,X_{i_k}]\lrcorner
X_{i_{k+1}}\lrcorner \cdots \lrcorner Z\lrcorner\Omega=0.
$$
Moreover, $L_Z\Omega=\textrm{div}(Z)\Omega$.
\end{proof}

\begin{corollary}
  Lie remarkable equations constructed by means of Proposition
  \ref{prop.sufficienza} or \ref{prop.sufficienza.2} possess first integrals of
  type \eqref{eq.first.int}.
\end{corollary}

\begin{remark}\label{rem.sufficient.numbers}
  Proposition \ref{prop.simm.first} says that if we have a sufficient number of
  point symmetries we can construct first integrals. In the case we have
  exactly $n$ symmetries, \eqref{eq.first.int} is a constant, so it is a
  trivial first integral. The above construction starts to produce non-trivial
  first integrals when we have at least $n+1$ symmetries.
\end{remark}

Below, as an example of computation, we use Proposition \ref{prop.simm.first}
to construct first integrals of equations \eqref{eq.family}. We recall that Lie
point symmetries of \eqref{eq.family} are \eqref{eq.symm.family} so that the
following integrals are obtained:
\begin{equation*}
  I_1=\frac{\det \left(
      \begin{array}{ccc}
        1 & u_x & \frac{1}{2}u_x+Ke^{-2x}u_x^3
        \\
        0 & 1 & 0
        \\
        1 & u & u_x
      \end{array}
    \right)}{\det \left(
      \begin{array}{ccc}
        1 & u_x & \frac{1}{2}u_x+Ke^{-2x}u_x^3
        \\
        0 & 1 & 0
        \\
        u & \frac{1}{2}u^2 & u_xu-u_x^2
      \end{array}
    \right)} = \frac{2Ke^{-2x}u_x^2-1}{u(2Ke^{-2x}u_x^2-1)+2u_x},
\end{equation*}
and
\begin{equation*}
  I_2=\frac{\det \left(
      \begin{array}{ccc}
        1 & u_x & \frac{1}{2}u_x+Ke^{-2x}u_x^3
        \\
        0 & 1 & 0
        \\
        1 & u & u_x
      \end{array}
    \right)}{\det \left(
      \begin{array}{ccc}
        1 & u_x & \frac{1}{2}u_x+Ke^{-2x}u_x^3
        \\
        1 & u & u_x
        \\
        u & \frac{1}{2}u^2 & u_xu-{u}_x^2
      \end{array}
    \right)} = 2\frac{2Ke^{-2x}{u}_x^2-1}{u^2(2Ke^{-2x}{u}_x^2-1)+4u_x(u-u_x)}.
\end{equation*}

\section{Proof of Theorem \ref{theo.main} and Theorem
  \ref{thm:systems}}\label{sec.proof}

Let us consider the euclidean space $\mathbb{R}^{1+m}$. From now on, we
interpret such a space as the space of $1$ independent variable and $m$
dependent variables, \emph{i.e.}, $\mathbb{R}^{1+m}=J^0(1,m)$. In this section,
by considering Propositions \ref{prop.chissa}, \ref{prop.sufficienza},
\ref{prop.sufficienza.2} and \ref{crit}, we construct Lie remarkable
ODEs by starting from a given Lie algebra $\mathfrak{s}$ of vector fields on
$\mathbb{R}^{1+m}$.  More precisely, we start by taking into account primitive
Lie algebras of vector fields on $\mathbb{R}^{2}$ and construct (when possible)
the corresponding associated Lie remarkable scalar ODEs. Then we consider the
isometric, special conformal, affine and projective algebra of $\R^{1+m}$ in
order to obtain Lie remarkable systems of ODEs.

\smallskip The dimension of the Lie algebras of vector fields we are
going to consider is a function of the number of the dependent variables, so
that, as a first step, we ask ourselves the following question: by using
Proposition \ref{prop.sufficienza} or \ref{prop.sufficienza.2}, which
combinations of integer values of $r$ and $m$ can produce a Lie remarkable
system of $m$ ODEs of order $r$ associated with the given Lie algebra of vector fields
$\mathfrak{s}$? We realize that such systems will be found as $m$-codimensional
submanifolds of the jet space $J^r(1,m)$ such that
\begin{equation}
  \dim \mathfrak{s} \geq \dim J^r(1,m)-m=rm+1.
\end{equation}
After this dimensional estimation, we obtain the Lie remarkable system
$\mathcal{E}$ associated with the considered Lie algebra of vector fields $\mathfrak{s}$ by
imposing on $\mathcal{E}$ the conditions of Proposition \ref{prop.sufficienza}
or \ref{prop.sufficienza.2}. In particular, in order to satisfy conditions of
Proposition \ref{prop.sufficienza}, we find the set
$\textrm{sing}_k(\mathfrak{s}^{(r)})$ of points of $J^r(1,m)$ where the rank of
the distribution associated with the prolonged Lie algebra $\mathfrak{s}^{(r)}$
of the chosen Lie algebra of vector fields $\mathfrak{s}$ is at most equal to $k$, where
$k\in\mathbb{N}$ is smaller than the maximal rank of $\mathfrak{s}^{(r)}$. The
Lie remarkable equation, whether it exists, is one of this singular set as the
vector fields in $\mathfrak{s}^{(r)}$ are tangent to the set
$\textrm{sing}_k(\mathfrak{s}^{(r)})$, provided that the latter is a
submanifold of $J^r(1,m)$ (see also Theorem \ref{th.Alek.Michor}.


The computation of the singular set leads to algebraic computations that we
solve with the help of computer algebra \cite{Relie}.  The rank condition of
Proposition \ref{prop.sufficienza} is enforced as a system of equations
obtained from the vanishing of determinants of all square submatrices (minors),
of suitable dimension, of the matrix of prolonged generators
$\mathcal{M}_{\mathfrak{s}^{(r)}}$ (see \eqref{eq.matrice.prolungamenti}). We
underline that, in the scalar case and for the considered Lie algebras, the
method described above reduces to that of Lie determinant, i.e. Lie remarkable
equations of order at most $r$ associated with an $(r+2)$-dimensional Lie
algebra of vector fields are described by the vanishing of the determinant of the $(r+2)\times
(r+2)$ matrix $\mathcal{M}_{\mathfrak{s}^{(r)}}$.

\subsection{Primitive Lie algebras and Lie remarkable scalar ODEs:
  proof of Theorem \ref{theo.main}}\label{sec.primitive}

Let us set $m=1$. Starting from the primitive (non-singular) Lie algebras of
vector fields of $\mathbb{R}^2$, we construct Lie remarkable scalar ODEs taking
Propositions \ref{prop.sufficienza}, \ref{prop.sufficienza.2} and
\ref{crit} into account. For the convenience of the reader, below we report the
list of the afore-mentioned Lie algebras, which can be found in \cite{Olver2}.
\begin{equation}\label{eq.primitive}
  \begin{tabular}{|c|l|l|}
    \hline
    & \text{\hspace{4cm} Generators} & \text{Structure} \\
    \hline
    \textbf{I} & $\partial_x,\partial_u,u\partial_x-x\partial_u+\alpha(x\partial_x+u\partial_u), \, \alpha\in\mathbb{R}$ & $\mathbb{R}\ltimes\mathbb{R}^2$ \\
    \hline
    \textbf{II} & $\partial_x,x\partial_x+u\partial_u,(x^2-u^2)\partial_x+2xu\partial_u$ & $\mathfrak{sl}(2)$ \\
    \hline
    \textbf{III} & $u\partial_x-x\partial_u, (1+x^2-u^2)\partial_x+2xu\partial_u,2xu\partial_x+(1-x^2+u^2)\partial_u$ & $\mathfrak{so}(3)$ \\
    \hline
    \textbf{IV} & $\partial_x,\partial_u,x\partial_x+u\partial_u,u\partial_x-x\partial_u$ & $\mathbb{R}^2\ltimes\mathbb{R}^2$\\
    \hline
    \textbf{V} & $\partial_x,\partial_u,x\partial_x-u\partial_u,u\partial_x,x\partial_u$ & $\mathfrak{sa}(2)$ \\
    \hline
    \textbf{VI} & $\partial_x,\partial_u,x\partial_x,u\partial_u,u\partial_x,x\partial_u$  & $\mathfrak{a}(2)$ \\
    \hline
    \textbf{VII} &  $\partial_x,\partial_u, u\partial_x-x\partial_u,x\partial_x+u\partial_u,
    (x^2-u^2)\partial_x+2xu\partial_u,
    2xu\partial_x+(u^2-x^2)\partial_u$ & $\mathfrak{so}(3,1)$ \\
    \hline
    \textbf{VIII} & $\partial_x,\partial_u,x\partial_x,u\partial_u,u\partial_x,x\partial_u,x^2\partial_x+xu\partial_u,xu\partial_x+u^2\partial_u$ & $\mathfrak{sl}(3)$ \\
    \hline
  \end{tabular}
\end{equation}

In this case, since the Lie algebra $\textbf{I}$ is $3$-dimensional,
Proposition \ref{prop.sufficienza} provides sufficient conditions for the
existence of a $1^{\textrm{st}}$ order Lie remarkable ODE associated to $\textbf{I}$,
whereas Proposition \ref{prop.sufficienza.2} for the existence of a $2^{\textrm{nd}}$
order one. A direct computation shows that such ODEs do not exist. In fact, the
matrix of $2$-prolongations of \textbf{I} is
\begin{equation*}
  \mathcal{M}_{\textbf{I}^{(2)}}=\left(
    \begin{array}{cccc}
      1 & 0 & 0 & 0
      \\
      0 & 1 & 0 & 0
      \\
      u+\alpha x & -x+\alpha u & -1-u_x^2 & -u_{xx}(3u_x+\alpha)
    \end{array}
  \right) =
  \left(
    \begin{array}{cc}
      &  0
      \\
      \mathcal{M}_{\textbf{I}^{(1)}} &  0
      \\
      & -u_{xx}(3u_x+\alpha)
    \end{array}
  \right),
\end{equation*}
so that the rank of the submatrix of $\mathcal{M}_{\textbf{I}^{(1)}}$ is
$3$-dimensional everywhere, which implies that there are no $1^{\textrm{st}}$ order Lie
remarkable ODEs associated with Lie algebra $\textbf{I}$. We cannot construct
Lie remarkable equation by means of Proposition \ref{prop.sufficienza.2} as the
necessary conditions contained in Remark \ref{rem.nec.cond.suff.2} are not
satisfied. Anyway, for our purposes, we write the only solution to system
$\mathcal{M}_{\textbf{I}^{(2)}}(F)=0$, up to functional dependency:
\begin{equation}\label{eq.diff.inv.1}
  F(u_x,u_{xx}):=\frac{e^{-\alpha\arctan(u_x)}u_{xx}}{(1+u_x^2)^{\frac{3}{2}}}\,,
  \,\,\alpha\in\mathbb{R},
\end{equation}
Proposition \ref{crit} is not satisfied. Indeed $F(u_x,u_{xx})$ is a differential
invariant of Lie algebra \textbf{I} and the most general $2^{\textrm{nd}}$ order scalar
ODEs admitting \textbf{I} as subalgebra of symmetries is $F(u_x,u_{xx})=k$,
$k\in\mathbb{R}$, but previous equations are not all equivalent each other as
for $k=0$ we have an equation belonging to the class of projective connections,
(i.e. equations $u_{xx}=P(x,u,u_x)$ where $P$ is a polynomial of third degree
in $u_x$), which is closed w.r.t. point transformations, whereas for $k\neq 0$
the equation does not belong to this class.  Now we answered a question posed
in Remark \ref{rem.controesempio}. To this aim, note that equations of the
family $D_x(F)=0$, which coincides with
$(1+u_x^2)u_{xxx}=u_{xx}^2(\alpha+3u_x)$, are not all equivalent each other
since for $\alpha=0$ we obtain a Lie remarkable equation with a $6$-dimensional
symmetry algebra (see also below, when discussing the case of Lie algebra
\textbf{VII}) whereas for $\alpha\neq 0$ we do not have any of these
properties. Taking into account that the only solution to system
$\mathcal{M}_{\textbf{I}^{(2)}}(F)=0$ is \eqref{eq.diff.inv.1}, we answered
negatively to the question posed in Remark \ref{rem.controesempio}.

Note that $ds={e^{\alpha\arctan(u_x)}}{(1+u_x^2)^{\frac{1}{2}}}dx$ is an
invariant $1$-form; a function of $F(u_x,u_{xx})$ and of its derivatives with
respect to $ds$ is, in its turn, a differential invariant. For $\alpha=0$ we
find the classical result of the euclidean geometry that any (euclidean)
differential invariant is a function of the curvature and of its
derivatives. Similar results can be obtained also for the other Lie algebras:
differential invariants of degree $0$ and $1$ play the same role as the
curvature and the arc-length element for the Klein geometries associated to
algebras \eqref{eq.primitive}. We do not insist here on this aspect even if it
is somehow related to the topic of this paper.

By using the same reasonings we adopted for the algebra \textbf{I}, we can
prove that there are no Lie remarkable equations associated with Lie algebras
\textbf{II} and \textbf{III}. In fact, according to Proposition
\ref{prop.sufficienza}, for dimensional reasons such equations should be of
$1^{\textrm{st}}$ order, but there are no hypersurfaces of $J^1(1,1)$ of order $1$
(i.e. $1^{\textrm{st}}$ order ODEs) where the rank of $\mathcal{M}_{\textbf{II}^{(1)}}$
(resp. $\mathcal{M}_{\textbf{III}^{(1)}}$) drops.  Both Lie algebras
\textbf{II} and \textbf{III} do not satisfy the hypotheses of Proposition
\ref{prop.sufficienza.2} as they do not satisfy necessary conditions of Remark
\ref{rem.nec.cond.suff.2}. In fact, the rank of both
$\mathcal{M}_{\textbf{II}^{(1)}}$ and $\mathcal{M}_{\textbf{III}^{(1)}}$ is
almost everywhere equal to $3$. For dimensional reason, this is the only case
to be considered. This concludes the proof of item \ref{eq.no.Lie.rem.scalar}
of Theorem \ref{theo.main}. Moreover, both Lie algebras \textbf{II} and
\textbf{III} do not satisfy Proposition \ref{crit}. Indeed, the unique $2^{\textrm{nd}}$
order differential invariant of \textbf{II} (resp. \textbf{III}) up to
functional dependence, is $(1+u_x^2)^{-\frac{3}{2}}(u_{xx}u+1+u_x^2)$
(resp. $(1+u_x^2)^{-\frac{3}{2}}(1+x^2+u^2)u_{xx}+2(1+u_x^2)^{-\frac{1}{2}}(u-xu_x)$)
and our assertion follows by applying the same reasoning as for the Lie algebra
\textbf{I}.

By continuing in our analysis, we see that the only Lie remarkable equation
associated with algebras \textbf{IV} or \textbf{V} is $u_{xx}=0$. In fact
\begin{equation}\label{eq.matrix.IV}
  \mathcal{M}_{\textbf{IV}^{(2)}}=\left(
    \begin{array}{cccc}
      1 & 0 & 0 & 0
      \\
      0 & 1 & 0 & 0
      \\
      x & u & 0 & -u_{xx}
      \\
      u & -x & -1-u_x^2 & -3u_xu_{xx}
    \end{array}
  \right)
\end{equation}
and the only hypersurface of $J^2(1,1)$ where the rank of \eqref{eq.matrix.IV}
is less than $4$ is $u_{xx}=0$. For dimensional reasons, according to
Proposition \ref{prop.sufficienza.2}, Lie algebra $\textbf{IV}$ could produce a
Lie remarkable ODE of order at most $3$, but necessary conditions contained in
Remark \ref{rem.nec.cond.suff.2} are not satisfied. In fact, in view of
\eqref{eq.matrix.IV}, the rank of $\mathcal{M}_{\textbf{IV}^{(r)}}$ is, almost
everywhere, equal to $r+2$ for $r\leq 2$. A similar reasoning applies also to
Lie algebra \textbf{V}. Furthermore, both Lie algebras \textbf{IV} and
\textbf{V} do not satisfy Proposition \ref{crit}. Indeed the unique $3^{\textrm{rd}}$
(resp. $4^{\textrm{th}}$) order differential invariant of \textbf{IV} (resp. \textbf{V})
up to functional dependence, is
$\mathcal{I}=u_{xx}^{-2}((1+u_x^2)u_{xxx}-3u_xu_{xx}^2)$
(resp. $\mathcal{J}=u_{xx}^{-\frac{8}{3}}(3u_{xx}u_{xxxx}-5u_{xxx}^2)$). Now,
equations $\mathcal{I}=k$ with $k\in\mathbb{R}$ cannot be all equivalent each
other as for $k=0$ we have we have a $6$-dimensional Lie symmetry algebra
whereas for $k\neq 0$ a $5$-dimensional one. Also equation $\mathcal{J}=k$ with
$k\in\mathbb{R}$ cannot be all equivalent as the prolongation of a point
transformation is always and affine transformation in all derivatives except
for the first ones in relation to which it is a projective transformation, so
that it is not possible to transform equation $\mathcal{J}=0$ into equation
$\mathcal{J}=k$ for any $k\neq 0$. Since the full symmetry algebra of
$u_{xx}=0$ is the projective algebra $\mathfrak{sl}(3)$, any subalgebra of
$\mathfrak{sl}(3)$ containing either the Lie algebra \textbf{IV} or \textbf{V}
leads to the Lie remarkable equation $u_{xx}=0$ (see also Proposition
\ref{prop.chissa}). For instance, starting from affine algebra \textbf{VI} we
obtain again the Lie remarkable equation $u_{xx}=0$.

By considering the Lie algebras \textbf{VI}, \textbf{VII} and \textbf{VIII} we
can construct, by means of Proposition \ref{prop.sufficienza} and following the
same reasoning adopted so far, respectively equation of items \ref{eq.GKS},
\ref{eq.circle.2d} and \ref{eq.conic} of Theorem \ref{theo.main}.  We note that
none of these Lie algebra satisfies necessary conditions of Remark
\ref{rem.nec.cond.suff.2}. Moreover, also Proposition \ref{crit} cannot apply: it
is sufficient to look at the differential invariants of such Lie algebras (for
instance one can consult the list provided in \cite{Olver2}) and arguing as in
the case of Lie algebra \textbf{V}.

\smallskip This concludes the proof of Theorem \ref{theo.main}.

\subsection{Lie remarkable systems of ODEs: proof of Theorem
  \ref{thm:systems}}\label{sec.systems}

Here we consider the isometry $\cI(\R^{3})$, affine $\cA(\R^{3})$, conformal
$\cC(\R^{3})$ and projective $\cP(\R^{3})$ Lie algebra of $\R^{3}$. For the
convenience of the reader we recall that
\[
\dim \cI(\R^{3})=6,\quad \dim\cA(\R^{3})=12,\quad \dim\cC(\R^{3})=10,\quad
\dim\cP(\R^{3})=15.
\]

In the case of the Lie algebra of infinitesimal isometries $\cI(\R^{1+m})$, we
prove that, for any dimension, system $\{u_{xx}^k=0\,,\,\,k=1,\dots, m\}$ is
associated with $\cI(\R^{1+m})$. This leads to some computations involving the
study of a system of PDEs, which can be solved by using some tricks (see next
section). In general, similar systems are very hard to treat, so that the
algorithmic procedure described in the beginning of Section \ref{sec.proof}
provides as efficient tools for describing Lie remarkable systems associated
with a given Lie algebra of vector fields. We follow the latter methods in
Section \ref{sec.aff.as.Lie.sym}, \ref{sec.sc.as.Lie.sym} and
\ref{sec.proj.as.Lie.sym}.

\subsubsection{Systems of ODEs associated with Lie algebra
  $\cI(\R^{1+m})$}\label{sec.iso.as.Lie.symm}

Generators of infinitesimal isometries of $\mathbb{R}^{1+m}$ are
\begin{equation}\label{eq.gen.iso}
  \pd{}{x}, \quad \pd{}{u^i}, \quad -u^i\pd{}{x} +x\pd{}{u^i}\,,
  \quad u^i\pd{}{u^j}-u^j\pd{}{u^i}\,,\quad i\neq j\,,\quad i,j=1\dots m.
\end{equation}

There are no $1^{\textrm{st}}$ order systems ODEs admitting \eqref{eq.gen.iso} as Lie
symmetry subalgebra. Indeed, the most general system of $1^{\textrm{st}}$ order ODEs
admitting $\partial_x$ and $\partial_{u^i}$ as Lie symmetries is of the form
$\{u^k_x=c^k\}_{k=1\dots m}$, $c^k\in\mathbb{R}$. By imposing that the third
set of vector fields \eqref{eq.gen.iso} are symmetries of the above system, we
obtain the equation $\delta^{ik}+c^ic^k=0\,\,\forall\,i,k$, so that if $i=k$ we
obtain the contradiction $({c^k})^2+1=0$.

\smallskip Going to higher order systems, the most general system of
$2^{\textrm{nd}}$ order ODEs admitting $\partial_x$ and $\partial_{u^i}$ as Lie symmetries
is of the form
\begin{equation}\label{eq.sys.iso}
  \{u^k_{xx}=F^k(u^1_x,\dots,u^m_x)\}\,, \quad k=1\dots m.
\end{equation}
By imposing that the remaining vector fields of \eqref{eq.gen.iso} are
symmetries of \eqref{eq.sys.iso}, we obtain the following system of PDEs:
\begin{equation}\label{eq.sys.iso.2}
  \left\{
    \begin{array}{ll}
      -F^k_{u^i_x}-\sum_{h=1}^m u^i_x u^h
      _x F^k_{u^j_x} + 2F^k u^i_x + F^iu^k_x=0, &
      \\
      & \qquad \forall\, i,j,k.
      \\
      -u^i_x F^j_{u^j_x} + u^j_x F^k_{u^i_x} + F^i\delta^{kj} - F^j\delta^{ki}=0, &
    \end{array}
  \right.
\end{equation}

By integrating the system formed by equations of the first line of
\eqref{eq.sys.iso.2} with $k=i$, we obtain that
\begin{equation}\label{eq.iso.more.gen}
  F^k=c^k\left(1+\sum_{j=1}^m(u^j_x)^2\right)^{\frac{3}{2}}.
\end{equation}

Let us observe that in the case $m=1$ the second set of equations of
\eqref{eq.sys.iso.2} is always satisfied, so that \eqref{eq.iso.more.gen}
reduces to \eqref{eq.diff.inv.1} with $\alpha=0$, which turns out the most
general scalar ODE (euclidean curvature equal to a constant) admitting the
euclidean algebra of $\mathbb{R}^2$ as a Lie symmetry subalgebra.

In the case $m\geq 2$, by substituting \eqref{eq.iso.more.gen} into the second
set of equations of \eqref{eq.sys.iso.2} with $k=i\neq j$ gives $c^j=0$ for
$j\neq k$. In view of arbitrariness of $k$, we obtain $c^k=0$ $\forall\, k$,
that is system of item \ref{eq.lines.nd} of Theorem \ref{thm:systems}. This
proves item \ref{eq.lines.nd} of Theorem~\ref{thm:systems}.

\subsubsection{Systems of ODEs associated with Lie algebra $\cA(\R^{3})$}
\label{sec.aff.as.Lie.sym}

The generators of this algebra are
\[
\pd{}{a}, \quad a\pd{}{b},
\]
for all $a, b\in \{x,u^i\}$. We can prove, by using the reasonings contained in
the beginning of Section \ref{sec.proof}, that when $r=5$ there exists the
following Lie remarkable equation ($u$ and $v$ denote the dependent variables):
\begin{align*}
  u_{xxxxx}=&(5( - u_{xx}^4v_{xxxx}^3 + 6u_{xx}^3u_{xxx}v_{xxx}v_{xxxx}^2 + 3
  u_{xx}^3u_{xxxx}v_{xx}v_{xxxx}^2 \\ & + 36u_{xx}^3u_{xxxx}v_{xxx}^2v_{xxxx} -
  6u_{xx}^2 u_{xxx}^2v_{xx}v_{xxxx}^2 - 72u_{xx}^2u_{xxx}^2v_{xxx}^2v_{xxxx} -
  \\ & 84u_{xx}^2u_{xxx}u_{xxxx}v_{xx}v_{xxx}v_{xxxx} +
  72u_{xx}^2u_{xxx}u_{xxxx}v_{xxx}^3 - 3u_{xx}^2u_{xxxx}^2v_{xx}^2v_{xxxx} \\ &
  - 36u_{xx}^2u_{xxxx}^2v_{xx}v_{xxx}^2 +
  144u_{xx}u_{xxx}^3v_{xx}v_{xxx}v_{xxxx} +
  48u_{xx}u_{xxx}^2u_{xxxx}v_{xx}^2v_{xxxx} \\ & -
  144u_{xx}u_{xxx}^2u_{xxxx}v_{xx}v_{xxx}^2 +
  78u_{xx}u_{xxx}u_{xxxx}^2v_{xx}^2v_{xxx} + u_{xx}u_{xxxx}^3v_{xx}^3 \\ & -
  72u_{xxx}^4v_{xx}^2v_{xxxx} + 72u_{xxx}^3u_{xxxx}v_{xx}^2v_{xxx} - 42
  u_{xxx}^2u_{xxxx}^2v_{xx}^3)) \\ & \big/ (144(u_{xx}^3v_{xxx}^3 -
  3u_{xx}^2u_{xxx}v_{xx}v_{xxx} ^2 + 3u_{xx}u_{xxx}^2v_{xx}^2v_{xxx} -
  u_{xxx}^3v_{xx}^3)),
  \\
  v_{xxxxx}=&(5( - u_{xx}^3v_{xx}v_{xxxx}^3 + 42u_{xx}^3v_{xxx}^2v_{xxxx}^2 -
  78u_{xx} ^2u_{xxx}v_{xx}v_{xxx}v_{xxxx}^2 \\ & -
  72u_{xx}^2u_{xxx}v_{xxx}^3v_{xxxx} + 3u_{xx}^2u_{xxxx}v_{xx}^2v_{xxxx}^2 -
  48u_{xx}^2u_{xxxx}v_{xx}v_{xxx}^2v_{xxxx} \\ & + 72u_{xx}^2u_{xxxx}v_{xxx}^4
  + 36u_{xx}u_{xxx}^2v_{xx}^2v_{xxxx}^2 + 144u_{xx}u_{xxx}^2v_{xx}
  v_{xxx}^2v_{xxxx} \\ & + 84u_{xx}u_{xxx}u_{xxxx}v_{xx}^2v_{xxx}v_{xxxx} -
  144u_{xx}u_{xxx} u_{xxxx}v_{xx}v_{xxx}^3 - 3u_{xx}u_{xxxx}^2v_{xx}^3v_{xxxx}
  \\ & + 6u_{xx}u_{xxxx}^2v_{xx}^2v_{xxx}^2 -
  72u_{xxx}^3v_{xx}^2v_{xxx}v_{xxxx} - 36u_{xxx}^2u_{xxxx}v_{xx}^3v_{xxxx} \\ &
  + 72u_{xxx}^2u_{xxxx}v_{xx}^2v_{xxx}^2 - 6u_{xxx}u_{xxxx}^2v_{xx}^3v_{xxx} +
  u_{xxxx}^3v_{xx}^4)) \\ & \big /(144(u_{xx}^3v_{xxx}^3 -
  3u_{xx}^2u_{xxx}v_{xx}v_{xxx}^2 + 3 u_{xx}u_{xxx}^2v_{xx}^2v_{xxx} -
  u_{xxx}^3v_{xx}^3)).
\end{align*}
A direct computation proves the following property.
\begin{proposition}
  The expressions that define the above two equations are differential
  invariants of the affine Lie algebra.
\end{proposition}

As it is well-known \cite{Olver2} the number of independent differential
invariants for space curves is $2$. In this sense the above system generalizes
the result in item \ref{eq.GKS} of Theorem \ref{theo.main}, where the affine
algebra uniquely determines the vanishing of the only differential invariant of
plane curves. It is reasonable to conjecture that an analogous result could hold in higher
dimensional euclidean spaces.

\subsubsection{Systems of ODEs associated with Lie algebra
  $\cC(\R^{3})$}\label{sec.sc.as.Lie.sym}

Conformal vector fields of the euclidean space $\mathbb{R}^n$ form a Lie
algebra of dimension $\frac{1}{2}(n+1)(n+2)$ except for the case $n=2$. In
fact, in such a case, the dimension of the Lie algebra of conformal vector
fields is infinite, as the latter are essentially identified with holomorphic
functions of one complex variable. More precisely, if $X_1+iX_2$ is a
holomorphic function, then $X_1\partial_x+X_2\partial_y$ is a conformal vector
field of the euclidean metric. If $n\geq 3$, the conformal algebra is formed by
translations, rotations, the dilatation and special conformal vector
fields. The local flow of such vector fields consists of special conformal
transformations: such a transformation can be understood as an inversion
followed by a translation and again followed by an inversion.  Since it is well
known that there are no systems of ODEs admitting an infinite dimensional Lie
algebra of point symmetries (provided that such algebra has no zero-dimensional
orbits), for the case $n=2$, instead of considering the full conformal Lie
algebra, we shall consider the Lie subalgebra formed by the $2$ translations,
the rotation, the dilation and the $2$ special conformal vector
fields. Let us introduce the notation $y^0=x$ and $y^i=u^i$. Then, in
any dimension, special conformal vector field are given by
\begin{equation*}
  \Xi_j=(2(y^j)^2-\|y\|^2)\partial_{y^j} + 2\sum_{i\neq j}y^jy^i\partial_{y^i},
\end{equation*}
and their local flow is
\begin{equation*} {y^i_0}'=\frac{y^i_0-\delta^i_j t
    \|y_0\|^2}{1-2ty^j_0+t^2\|y_0\|^2}.
\end{equation*}
Special conformal transformations preserve circles of $\mathbb{R}^n$. By
adopting the method described in the beginning of Section \ref{sec.proof}, we
proved by a straightforward computation that the only system of ODEs admitting
 $\cC(\mathbb{R}^3)$ as Lie symmetry subalgebra is
system of item \ref{eq.circle.nd} of Theorem \ref{thm:systems}.

We observe that in a recent paper \cite{CDT} a system of ODEs is
  associated with the conformal symmetry algebra of an anti-self-dual conformal
  structure defined by a pseudo-Riemannian metric which is a solution of the
  Plebanski equation. If the dimension of the conformal symmetry algebra is
  high enough, the associated system of ODEs is Lie remarkable (like in
  Example~1 of \cite{CDT}), and can be determined by our methods, as an alternative to the
  method proposed in \cite{CDT}.

\subsubsection{Systems od ODEs associated with Lie algebra
  $\cP(\R^{3})$}\label{sec.proj.as.Lie.sym}

The generators of Lie algebra $\cP(\R^{3})$ are
\[
\pd{}{a}, \quad a\pd{}{b}, \quad a\left(x\pd{}{x} + \sum_{i=1}^m
  u^i\pd{}{u^i}\right),
\]
for all $a,b\in\{x,u^i\}$. We can prove, by using the reasonings contained in
the beginning of Section \ref{sec.proof}, that when $r=6$ there exists the
following Lie remarkable equation ($u$ and $v$ denote the dependent variables):
\begin{align*}
  u_{xxxxxx}=&( - 144u_{xx}^5v_{xxx}^2v_{xxxxx}^2 + 360u_{xx}^5v_{xxx}v_{xxxx}
  ^2v_{xxxxx} - 225u_{xx}^5v_{xxxx}^4 +
  \\
  & 288u_{xx}^4u_{xxx}v_{xx}v_{xxx}v_{xxxxx}^2 -
  360u_{xx}^4u_{xxx}v_{xx}v_{xxxx}^2v_{xxxxx}
  \\
  & - 480u_{xx}^4u_{xxx}v_{xxx}^2v_{xxxx} v_{xxxxx} +
  600u_{xx}^4u_{xxx}v_{xxx}v_{xxxx}^3
  \\
  &- 720u_{xx}^4u_{xxxx}v_{xx}v_{xxx} v_{xxxx}v_{xxxxx} +
  900u_{xx}^4u_{xxxx}v_{xx}v_{xxxx}^3 \\ &+960u_{xx}^4u_{xxxx}v_{xxx}^3
  v_{xxxxx} - 1200u_{xx}^4u_{xxxx}v_{xxx}^2v_{xxxx}^2 \\ & +
  288u_{xx}^4u_{xxxxx}v_{xx} v_{xxx}^2v_{xxxxx} -
  360u_{xx}^4u_{xxxxx}v_{xx}v_{xxx}v_{xxxx}^2 \\ &+ 240u_{xx}^4
  u_{xxxxx}v_{xxx}^3v_{xxxx} - 144u_{xx}^3u_{xxx}^2v_{xx}^2v_{xxxxx}^2 \\
  &+ 960u_{xx}^3 u_{xxx}^2v_{xx}v_{xxx}v_{xxxx}v_{xxxxx} -
  600u_{xx}^3u_{xxx}^2v_{xx}v_{xxxx}^3 \\ & - 640
  u_{xx}^3u_{xxx}^2v_{xxx}^3v_{xxxxx} + 400u_{xx}^3u_{xxx}^2v_{xxx}^2v_{xxxx}^2
  \\ &+ 720u_{xx}^3u_{xxx}u_{xxxx}v_{xx}^2v_{xxxx}v_{xxxxx} -
  2400u_{xx}^3u_{xxx}u_{xxxx}v_{xx} v_{xxx}^2v_{xxxxx} \\ &+
  600u_{xx}^3u_{xxx}u_{xxxx}v_{xx}v_{xxx}v_{xxxx}^2 + 400u_{xx}^3
  u_{xxx}u_{xxxx}v_{xxx}^3v_{xxxx} \\ &-
  576u_{xx}^3u_{xxx}u_{xxxxx}v_{xx}^2v_{xxx}v_{xxxxx} +
  360u_{xx}^3u_{xxx}u_{xxxxx}v_{xx}^2v_{xxxx}^2 \\ &-
  240u_{xx}^3u_{xxx}u_{xxxxx}v_{xx} v_{xxx}^2v_{xxxx} +
  640u_{xx}^3u_{xxx}u_{xxxxx}v_{xxx}^4 \\ &+ 360u_{xx}^3u_{xxxx}^2
  v_{xx}^2v_{xxx}v_{xxxxx} - 1350u_{xx}^3u_{xxxx}^2v_{xx}^2v_{xxxx}^2 \\ & +
  2400u_{xx}^3 u_{xxxx}^2v_{xx}v_{xxx}^2v_{xxxx} -
  800u_{xx}^3u_{xxxx}^2v_{xxx}^4 \\ &+ 720u_{xx}^3
  u_{xxxx}u_{xxxxx}v_{xx}^2v_{xxx}v_{xxxx} -
  1200u_{xx}^3u_{xxxx}u_{xxxxx}v_{xx}v_{xxx}^3 \\ & -
  144u_{xx}^3u_{xxxxx}^2v_{xx}^2v_{xxx}^2 -
  480u_{xx}^2u_{xxx}^3v_{xx}^2v_{xxxx} v_{xxxxx} \\ & +
  1920u_{xx}^2u_{xxx}^3v_{xx}v_{xxx}^2v_{xxxxx} -
  800u_{xx}^2u_{xxx}^3v_{xx}v_{xxx}v_{xxxx}^2 \\ & +
  1920u_{xx}^2u_{xxx}^2u_{xxxx}v_{xx}^2v_{xxx}v_{xxxxx} +
  600u_{xx}^2u_{xxx}^2u_{xxxx}v_{xx}^2v_{xxxx}^2 \\ & -
  2000u_{xx}^2u_{xxx}^2u_{xxxx}v_{xx} v_{xxx}^2v_{xxxx} +
  288u_{xx}^2u_{xxx}^2u_{xxxxx}v_{xx}^3v_{xxxxx} \\ & - 240u_{xx}^2
  u_{xxx}^2u_{xxxxx}v_{xx}^2v_{xxx}v_{xxxx} -
  1920u_{xx}^2u_{xxx}^2u_{xxxxx}v_{xx}v_{xxx} ^3 \\ & -
  360u_{xx}^2u_{xxx}u_{xxxx}^2v_{xx}^3v_{xxxxx} - 3000u_{xx}^2u_{xxx}u_{xxxx}^2
  v_{xx}^2v_{xxx}v_{xxxx} \\ & + 2800u_{xx}^2u_{xxx}u_{xxxx}^2v_{xx}v_{xxx}^3 -
  720u_{xx}^2u_{xxx}u_{xxxx}u_{xxxxx}v_{xx}^3v_{xxxx} \\ & +
  3120u_{xx}^2u_{xxx}u_{xxxx}u_{xxxxx}v_{xx}^2 v_{xxx}^2 +
  288u_{xx}^2u_{xxx}u_{xxxxx}^2v_{xx}^3v_{xxx} \\ & + 900u_{xx}^2u_{xxxx}^3
  v_{xx}^3v_{xxxx} - 1200u_{xx}^2u_{xxxx}^3v_{xx}^2v_{xxx}^2 \\ &-
  360u_{xx}^2u_{xxxx}^2 u_{xxxxx}v_{xx}^3v_{xxx} -
  1920u_{xx}u_{xxx}^4v_{xx}^2v_{xxx}v_{xxxxx} \\ & + 400u_{xx}
  u_{xxx}^4v_{xx}^2v_{xxxx}^2 - 480u_{xx}u_{xxx}^3u_{xxxx}v_{xx}^3v_{xxxxx} \\
  & + 2800 u_{xx}u_{xxx}^3u_{xxxx}v_{xx}^2v_{xxx}v_{xxxx} +
  240u_{xx}u_{xxx}^3u_{xxxxx}v_{xx}^3 v_{xxxx} \\ & +
  1920u_{xx}u_{xxx}^3u_{xxxxx}v_{xx}^2v_{xxx}^2 + 600u_{xx}u_{xxx}^2u_{xxxx}
  ^2v_{xx}^3v_{xxxx} \\ & - 3200u_{xx}u_{xxx}^2u_{xxxx}^2v_{xx}^2v_{xxx}^2 -
  2640u_{xx} u_{xxx}^2u_{xxxx}u_{xxxxx}v_{xx}^3v_{xxx} \\ &-
  144u_{xx}u_{xxx}^2u_{xxxxx}^2v_{xx}^4 +
  1800u_{xx}u_{xxx}u_{xxxx}^3v_{xx}^3v_{xxx} \\ & +
  360u_{xx}u_{xxx}u_{xxxx}^2u_{xxxxx}v_{xx} ^4 - 225u_{xx}u_{xxxx}^4v_{xx}^4 +
  640u_{xxx}^5v_{xx}^3v_{xxxxx} \\ & - 1200u_{xxx}^4 u_{xxxx}v_{xx}^3v_{xxxx} -
  640u_{xxx}^4u_{xxxxx}v_{xx}^3v_{xxx} + 1200u_{xxx}^3u_{xxxx}^2v_{xx}^3v_{xxx}
  \\ & + 720u_{xxx}^3u_{xxxx}u_{xxxxx}v_{xx}^4 - 600u_{xxx}^2
  u_{xxxx}^3v_{xx}^4) \big / (160u_{xx}^4v_{xxx}^4 - \\ &
  640u_{xx}^3u_{xxx}v_{xx}v_{xxx}^3 + 960 u_{xx}^2u_{xxx}^2v_{xx}^2v_{xxx}^2 -
  640u_{xx}u_{xxx}^3v_{xx}^3v_{xxx} + 160u_{xxx} ^4v_{xx}^4),
  \\
  v_{xxxxxx}=&( - 144u_{xx}^4v_{xx}v_{xxx}^2v_{xxxxx}^2 +
  360u_{xx}^4v_{xx}v_{xxx} v_{xxxx}^2v_{xxxxx} \\ & -
  225u_{xx}^4v_{xx}v_{xxxx}^4 + 720u_{xx}^4v_{xxx}^3v_{xxxx} v_{xxxxx} \\ & -
  600u_{xx}^4v_{xxx}^2v_{xxxx}^3 + 288u_{xx}^3u_{xxx}v_{xx}^2v_{xxx}
  v_{xxxxx}^2 \\ & - 360u_{xx}^3u_{xxx}v_{xx}^2v_{xxxx}^2v_{xxxxx} -
  2640u_{xx}^3u_{xxx}v_{xx}v_{xxx}^2v_{xxxx}v_{xxxxx} \\ & +
  1800u_{xx}^3u_{xxx}v_{xx}v_{xxx}v_{xxxx}^3 - 640
  u_{xx}^3u_{xxx}v_{xxx}^4v_{xxxxx} \\ & +
  1200u_{xx}^3u_{xxx}v_{xxx}^3v_{xxxx}^2 - 720
  u_{xx}^3u_{xxxx}v_{xx}^2v_{xxx}v_{xxxx}v_{xxxxx} \\ & +
  900u_{xx}^3u_{xxxx}v_{xx}^2v_{xxxx} ^3 +
  240u_{xx}^3u_{xxxx}v_{xx}v_{xxx}^3v_{xxxxx} \\ & +
  600u_{xx}^3u_{xxxx}v_{xx}v_{xxx}^ 2v_{xxxx}^2 -
  1200u_{xx}^3u_{xxxx}v_{xxx}^4v_{xxxx} \\ & +
  288u_{xx}^3u_{xxxxx}v_{xx}^2v_{xxx}^2v_{xxxxx} -
  360u_{xx}^3u_{xxxxx}v_{xx}^2v_{xxx}v_{xxxx}^2 \\ & - 480u_{xx}^3
  u_{xxxxx}v_{xx}v_{xxx}^3v_{xxxx} + 640u_{xx}^3u_{xxxxx}v_{xxx}^5 \\ & -
  144u_{xx}^2u_{xxx} ^2v_{xx}^3v_{xxxxx}^2 +
  3120u_{xx}^2u_{xxx}^2v_{xx}^2v_{xxx}v_{xxxx}v_{xxxxx} \\ & -
  1200u_{xx}^2u_{xxx}^2v_{xx}^2v_{xxxx}^3 +
  1920u_{xx}^2u_{xxx}^2v_{xx}v_{xxx}^3 v_{xxxxx} \\ & -
  3200u_{xx}^2u_{xxx}^2v_{xx}v_{xxx}^2v_{xxxx}^2 +
  720u_{xx}^2u_{xxx}u_{xxxx}v_{xx}^3v_{xxxx}v_{xxxxx} \\ & -
  240u_{xx}^2u_{xxx}u_{xxxx}v_{xx}^2v_{xxx}^2 v_{xxxxx} -
  3000u_{xx}^2u_{xxx}u_{xxxx}v_{xx}^2v_{xxx}v_{xxxx}^2 \\ & +
  2800u_{xx}^2u_{xxx}u_{xxxx}v_{xx}v_{xxx}^3v_{xxxx} -
  576u_{xx}^2u_{xxx}u_{xxxxx}v_{xx}^3v_{xxx}v_{xxxxx} \\ & +
  360u_{xx}^2u_{xxx}u_{xxxxx}v_{xx}^3v_{xxxx}^2 +
  1920u_{xx}^2u_{xxx}u_{xxxxx}v_{xx}^ 2v_{xxx}^2v_{xxxx} \\ & -
  1920u_{xx}^2u_{xxx}u_{xxxxx}v_{xx}v_{xxx}^4 +
  360u_{xx}^2u_{xxxx}^2v_{xx}^3v_{xxx}v_{xxxxx} \\ & -
  1350u_{xx}^2u_{xxxx}^2v_{xx}^3v_{xxxx}^2 + 600
  u_{xx}^2u_{xxxx}^2v_{xx}^2v_{xxx}^2v_{xxxx} \\ & +
  400u_{xx}^2u_{xxxx}^2v_{xx}v_{xxx}^4 +
  720u_{xx}^2u_{xxxx}u_{xxxxx}v_{xx}^3v_{xxx}v_{xxxx} \\ & -
  480u_{xx}^2u_{xxxx}u_{xxxxx} v_{xx}^2v_{xxx}^3 -
  144u_{xx}^2u_{xxxxx}^2v_{xx}^3v_{xxx}^2 \\ & - 1200u_{xx}u_{xxx}^3
  v_{xx}^3v_{xxxx}v_{xxxxx} - 1920u_{xx}u_{xxx}^3v_{xx}^2v_{xxx}^2v_{xxxxx} \\
  & + 2800u_{xx} u_{xxx}^3v_{xx}^2v_{xxx}v_{xxxx}^2 -
  240u_{xx}u_{xxx}^2u_{xxxx}v_{xx}^3v_{xxx} v_{xxxxx} \\ & +
  2400u_{xx}u_{xxx}^2u_{xxxx}v_{xx}^3v_{xxxx}^2 - 2000u_{xx}u_{xxx}^2
  u_{xxxx}v_{xx}^2v_{xxx}^2v_{xxxx} \\ & +
  288u_{xx}u_{xxx}^2u_{xxxxx}v_{xx}^4v_{xxxxx} -
  2400u_{xx}u_{xxx}^2u_{xxxxx}v_{xx}^3v_{xxx}v_{xxxx} \\ & +
  1920u_{xx}u_{xxx}^2u_{xxxxx} v_{xx}^2v_{xxx}^3 -
  360u_{xx}u_{xxx}u_{xxxx}^2v_{xx}^4v_{xxxxx} \\ & +
  600u_{xx}u_{xxx}u_{xxxx}^2v_{xx}^3v_{xxx}v_{xxxx} -
  800u_{xx}u_{xxx}u_{xxxx}^2v_{xx}^2v_{xxx}^3 \\ & -
  720u_{xx}u_{xxx}u_{xxxx}u_{xxxxx}v_{xx}^4v_{xxxx} +
  960u_{xx}u_{xxx}u_{xxxx}u_{xxxxx}v_{xx}^ 3v_{xxx}^2 \\ & +
  288u_{xx}u_{xxx}u_{xxxxx}^2v_{xx}^4v_{xxx} + 900u_{xx}u_{xxxx}^3v_{xx}^
  4v_{xxxx} \\ & - 600u_{xx}u_{xxxx}^3v_{xx}^3v_{xxx}^2 -
  360u_{xx}u_{xxxx}^2u_{xxxxx}v_{xx} ^4v_{xxx} \\ & +
  640u_{xxx}^4v_{xx}^3v_{xxx}v_{xxxxx} - 800u_{xxx}^4v_{xx}^3v_{xxxx}^2 +
  240u_{xxx}^3u_{xxxx}v_{xx}^4v_{xxxxx} \\ & +
  400u_{xxx}^3u_{xxxx}v_{xx}^3v_{xxx}v_{xxxx} +
  960u_{xxx}^3u_{xxxxx}v_{xx}^4v_{xxxx} -
  640u_{xxx}^3u_{xxxxx}v_{xx}^3v_{xxx}^2 \\ & -
  1200u_{xxx}^2u_{xxxx}^2v_{xx}^4v_{xxxx} +
  400u_{xxx}^2u_{xxxx}^2v_{xx}^3v_{xxx}^ 2 \\ & -
  480u_{xxx}^2u_{xxxx}u_{xxxxx}v_{xx}^4v_{xxx} -
  144u_{xxx}^2u_{xxxxx}^2v_{xx}^5 \\ & + 600u_{xxx}u_{xxxx}^3v_{xx}^4v_{xxx} +
  360u_{xxx}u_{xxxx}^2u_{xxxxx}v_{xx}^5 - 225 u_{xxxx}^4v_{xx}^5) \big/
  (160u_{xx}^4v_{xxx}^4 \\ & - 640u_{xx}^3u_{xxx}v_{xx}v_{xxx}^3 + 960
  u_{xx}^2u_{xxx}^2v_{xx}^2v_{xxx}^2 - 640u_{xx}u_{xxx}^3v_{xx}^3v_{xxx} +
  160u_{xxx} ^4v_{xx}^4).
\end{align*}

\textbf{Acknowledgments.} Work supported by G.N.F.M. and by G.N.S.A.G.A of
I.N.d.A.M. We thank Paola Morando of the University of Milan for stimulating
discussions.


\end{document}